\documentclass[12pt,a4paper]{article}
\usepackage[top=30mm, left=25mm, right=25mm, bottom=30mm]{geometry}

\usepackage[colorlinks=true, linkcolor=blue, urlcolor=blue, citecolor=red]{hyperref}

\usepackage{graphicx} 
\usepackage[T1]{fontenc}
\usepackage{imakeidx}
\makeindex[columns=1, title=Alphabetical Index, intoc]
\usepackage{amsmath}
\usepackage{amssymb}
\usepackage{tcolorbox}
\usepackage{verbatim}
\usepackage{xcolor}
\usepackage{listings}
\usepackage[colorlinks = true]{hyperref}
\usepackage[utf8]{inputenc}
\usepackage{amsmath}
\usepackage{amsfonts}
\usepackage{amsthm}
\usepackage{amsmath, amsfonts, amssymb, amsthm}
\usepackage{tikz}

\newtheorem{theorem}{Theorem}[section]

\newtheorem{corollary}[theorem]{Corollary}

\theoremstyle{definition}

\newtcolorbox{mytheorem}[2][]{colback=olive!5!white, colframe=blue!75!black,
fonttitle=\bfseries, title=#2, #1}

\usepackage{mathtools}

\title{Cartesian Prime Graphs and Cospectral Families}

\author{Abhinav Bitragunta,
    Hareshkumar Jadav,
    Ranveer Singh\\
    \small Department of Computer Science and Engineering,\\
    \small Indian Institute of Technology Indore, India
}

\date{}

\begin{document}

\maketitle

\begin{abstract}
We introduce a method for constructing larger families of connected cospectral graphs from two given cospectral families of sizes $p$ and $q$. The resulting family size depends on the Cartesian primality of the input graphs and can be one of $pq$, $p + q - 1$, or $\max(p, q)$, based on the strictness of the applied conditions. Under the strictest condition, our method generates $O(p^3q^3)$ new cospectral triplets, while the more relaxed conditions yield $\varOmega(pq^3 + qp^3)$ such triplets. We also use the existence of specific cospectral families to establish that of larger ones.

\end{abstract}

\section{Introduction}

One of the central questions in algebraic graph theory is whether the spectrum uniquely determines the structure of a graph or not. The existence of cospectral non-isomorphic graphs shows that the answer is often no, which motivates the problem of finding cospectral graphs. The notation $G \cong H$ denotes that $G$ and $H$ are isomorphic, while $G \ncong H$ indicates that they are non-isomorphic. Graphs $G$ and $H$ are said to be \textit{cospectral} with respect to some matrix if they have the same spectra on that matrix representation and $G \ncong H$. In this paper we focus on the adjacency and Laplacian matrix representation of graphs. Cospectral graphs must hence have the same number of vertices.

A widely used method to construct cospectral graphs is \textit{Godsil--McKay switching}~\cite{Godsil}. A variety of other constructions such as Schwenk’s coalescence method for trees, constructions using regular rational orthogonal matrices, and generalized switching techniques have been developed to produce both finite and infinite families of cospectral graphs~\cite{Abiad, Bapat, Haythorpe}. These methods not only deepen our understanding of the limitations of spectral invariants but also reveal rich algebraic and combinatorial structures underlying cospectrality. For further details refer to an interesting survey on this topic, see \cite{Haemers}. In this work, we develop a method to construct cospectral families of various sizes. We leverage Cartesian primality to construct such families.

The \textit{Cartesian product} of $G$ and $H$ is denoted by $G \square H$. We say that $G_1 \square G_2 \square \cdots \square G_k$ is a \textit{Cartesian factorization} of a graph $G$ if and only if $G \cong G_1 \square G_2 \square \cdots \square G_k$, and $G_1, \ldots, G_k$ are called the \textit{Cartesian factors} (or simply, factors) of $G$. A graph $G$ is said to be a \textit{Cartesian prime} if its only nontrivial factor is $G$ itself. If $G_1, \ldots, G_k$ are all prime, then the factorization is called a \textit{Cartesian prime factorization}. A fundamental result proved in \cite{Sabidussi} states that connected graphs have a unique Cartesian prime factorization. The graphs $G$ and $H$ are said to be \textit{coprime} if they share no prime factors.

Some well-known results on the Cartesian product of graphs are stated here \cite{Imrich}. The Cartesian product is commutative up to isomorphism. If $\lambda$ is an eigenvalue of $G$ and $\mu$ an eigenvalue of $H$, then $\lambda + \mu$ is an eigenvalue of $G \square H$, with multiplicities added. For non-trivial graphs $G_1, G_2$ and $H$, we have $G_1 \square H \cong G_2 \square H$ if and only if $G_1 \cong G_2$. The Cartesian Product $G \square H$ is connected if and only if both $G$ and $H$ are connected.

\section{Main Result}

In this section, we present a method to construct cospectral families of size $pq$ given two existing cospectral graph families of sizes $p$ and $q$. Our main result, with specific construction conditions, is detailed in Theorem~\ref{thm:main_const}.

\begin{theorem}
    \label{thm:main_const}
    Let $\mathcal{G} = \{G_1, G_2, \ldots, G_p\}$ and $\mathcal{H} = \{H_1, H_2, \ldots, H_q\}$ be families of connected, mutually cospectral graphs such that the spectra of graphs in $\mathcal{G}$ differ from those in $\mathcal{H}$. Suppose at least one of the following conditions holds:
\begin{enumerate}
    \item Every graph in \(\mathcal{G} \cup \mathcal{H}\) is Cartesian prime.
    \item \(\gcd(|V(G)|, |V(H)|) = 1\) for all $G \in \mathcal{G}$ and $H \in \mathcal{H}$.
    \item \label{cond:general} No Cartesian prime graph $P$ is a common factor of any $G \in \mathcal{G}$ and $H \in \mathcal{H}$ (i.e., the families share no prime factors)
    \end{enumerate}
    Then the family $\mathcal{F} = \{F_{ij} \mid F_{ij} = G_i \square H_j,\ 1 \leq i \leq p,\ 1 \leq j \leq q\}$ forms a connected cospectral family.
\end{theorem}

\begin{proof}
Since all graphs in \(\mathcal{G}\) and \(\mathcal{H}\) are mutually cospectral, it follows that every $F_{ij}$ has the same spectrum. It remains to show that no two elements of $\mathcal{F}$ are isomorphic.
For distinct elements $F_{ab}, F_{cd} \in \mathcal{F}$, there are three cases. The first is when $a = c$, where we have $H_b \ncong H_d \Rightarrow F_{ab} \ncong F_{cd}$. The second is when $b = d$ whose proof of non-isomorphism is similar. To prove the third case, we assume otherwise i.e., $F_{ab} \cong F_{cd}$ and arrive at a contradiction.

For Condition 1, we know that $F_{ab}$ and $ F_{cd}$ are connected since their factors are, and since they have the same unique factorization, it follows that $ G_a \cong G_c \text{ or } G_a \cong H_d $. The first contradicts the mutual non-isomorphism in $\mathcal{G}$, and the second contradicts the fact that the graphs in $\mathcal{G}$ and those in $\mathcal{H}$ have different spectra. Thus, Condition 1 ensures $\mathcal{F}$ is a cospectral family.

For Condition 2, suppose $G_a$ and $G_c$ have a common factor graph $\omega$, for which $G_a \cong \omega \square A$ and $G_c \cong \omega \square C$ and $\omega$ is maximal in the sense that the graphs $A$ and $C$ are coprime. Note that $\omega$ may be trivial, but neither $A$ nor $C$ can be trivial as otherwise, $G_a \cong G_c$. Similarly, write $H_b \cong \varOmega \square B$ and $H_d \cong \varOmega  \square D$. Since $F_{ab} \cong F_{cd}$ by the assumption, we get
\begin{align*}
        (\omega \square A) \square (\varOmega \square B) &\cong (\omega \square C) \square (\varOmega \square D) \\
        (\omega \square \varOmega) \square (A \square B) &\cong (\omega \square \varOmega) \square (C \square D) \\
        A \square B &\cong C \square D.
\end{align*}
Since $A,C$ and $B,D$ are pairwise coprime, we get $A \cong D$ and $B \cong C$. This implies $|V(G_a)| = |V(\omega)||V(A)|$ and $|V(H_d)| = |V(\Omega)||V(D)|$ share the common factor $|V(A)|=|V(D)| > 1$, contradicting $\gcd(|V(G_a)|, |V(H_d)|) = 1$. Thus, Condition 2 ensures $\mathcal{F}$ is a cospectral family.

For Condition 3, if \(G_a\) and \(H_d\) share no prime factors, their factorizations must be the same as \(G_c\) and \(H_b\), respectively. Therefore, any isomorphism $F_{ab} \cong F_{cd}$ would require $G_a \cong G_c$ and $H_b \cong H_d$, which contradicts their definition.

Thus, \(\mathcal{F}\) is a cospectral family of size \(pq\) under any of the three conditions. 
\end{proof}

The above result gives a connected cospectral family of size $pq$. The next result demonstrates that upon relaxing the third condition, a smaller connected cospectral family, of size $p+q-1$ can be constructed.

\begin{theorem}
    \label{thm:loosen} 
    Let $\mathcal{G} = \{G_1, G_2, \ldots, G_p\}$ and $\mathcal{H} = \{H_1, H_2, \ldots, H_q\}$ be families of connected, mutually cospectral graphs such that the spectra of graphs in $\mathcal{G}$ differ from those in $\mathcal{H}$. If there exist $G \in \mathcal{G}$ and $H \in \mathcal{H}$ such that $G$ and $H$ are coprime, then the largest mutually non-isomorphic subset $\mathcal{F''} \subseteq \mathcal{F}$ has size at least $p + q - 1$.
\end{theorem}

\begin{proof}
    Without loss of generality, suppose that $p \ge q$, and that $G_1$ and $H_1$ are coprime. The graph family  $\{F_{11}, F_{21}, \ldots, F_{p1}\}$ is a set of mutually cospectral connected graphs because if not, then for some $b \ne d$ we have $F_{b1} \cong F_{d1} \Rightarrow G_b \cong G_d$. Since $G_1$ and $H_1$ are coprime, we have $F_{1b} \ncong F_{c1}$ for any nontrivial choice of $b$ and $c$. Therefore, we can extend the cospectral family to 
    \[
    \mathcal{F''} = \{F_{11}, F_{21}, \ldots, F_{p1}, F_{12}, \ldots, F_{1q}\},
    \]
    and the size of this family is $p + q - 1$.
\end{proof}

We will count the number of new triplets of cospectral graphs from the above construction.

\begin{corollary}
    The construction defined in Theorem~\ref{thm:loosen} has at least
    \[
    (p-1)(q-1) + (p-1)\binom{q}{3} + (q-1)\binom{p}{3} + \binom{p+q-1}{3}
    \]
    newly generated connected cospectral triplets.
\end{corollary}

\begin{proof}
    Consider the subset $\mathcal{F}'' \subseteq \mathcal{F}$ of size $p + q - 1$ as defined in Theorem~\ref{thm:loosen}. First, for each pair $F_{a1}, F_{1b} \in \mathcal{F}''$ with $a \ne 1$ and $b \ne 1$, the third graph $F_{ab}$ lies in $\mathcal{F} \setminus \mathcal{F}''$. This has $(p - 1)(q - 1)$ triplets of the form $\{F_{a1}, F_{1b}, F_{ab}\}$. Second, all 3-element subsets of $\mathcal{F}''$ form connected cospectral triplets, contributing $\binom{p + q - 1}{3}$ additional triplets. Third, for each $j > 1$, the set $\{F_{1j}, F_{2j}, \ldots, F_{pj}\}$ is a family of $p$ mutually cospectral graphs. From each such family (of which there are $q - 1$), we can choose $\binom{p}{3}$ distinct triplets. This gives an additional $(q - 1)\binom{p}{3}$ triplets. Similarly, for each $i > 1$, the set $\{F_{i1}, F_{i2}, \ldots, F_{iq}\}$ contributes $\binom{q}{3}$ triplets. With $(p - 1)$ such families, this accounts for $(p - 1)\binom{q}{3}$ more.
    Summing all contributions yields the total number of newly generated connected cospectral triplets.
\end{proof}

A simpler lower bound of $\max(p,q)$ is obtained on dropping all the three conditions. To see this, observe that one can choose the larger of $\{F_{11}, F_{12}, \ldots, F_{1q}\}$ and $\{F_{11}, F_{21}, \ldots, F_{p1}\}$ to set the lower bound.

Note that Condition~3 is the most general among the three. If Condition~3 fails, then both Conditions~1 and~2 necessarily fail as well. Conversely, the truth of either Condition~1 or Condition~2 implies the truth of Condition~3. The following theorem establishes the existence of cospectral families of a specified size on a given number of vertices.

\begin{theorem}
    \label{thm:exist_family}
    If there exists a connected cospectral family of \( p \) mutually coprime graphs on \( n \) vertices, then for every integer \( k \ge 1 \), there exists a connected cospectral family of size \( \binom{k + p - 1}{k} \) on \( n^k \) vertices.
\end{theorem}

\begin{proof}
    Let \( \mathcal{U} = \{U_1, U_2, \ldots, U_p\} \) be such a family, where each \( U_i \) is connected, mutually cospectral, and pairwise coprime, and let \( |V(U_i)| = n \). The number of unique Cartesian prime factorizations of length $k$ can be seen as the distribution problem 
    $$k = e_1 + e_2 + \cdots + e_p,$$ 
    where each $e_i \ge 0$ denotes the multiplicity of $U_i$ in the factorization. The number of such factorizations is \( \binom{k + p - 1}{k} \) \cite{Erickson}, and the required family of graphs is the one obtained by including all such factorizations. 
    
    Each graph in this collection is connected and has \( n^k \) vertices, since the Cartesian product of graphs multiplies their vertex counts. Moreover, the mutual coprimality of the graphs in \( \mathcal{U} \) guarantees that each factorization corresponds to a unique (up to isomorphism) graph. Therefore, the constructed family consists of \( \binom{k + p - 1}{k} \) mutually non-isomorphic, connected, cospectral graphs on \( n^k \) vertices.
\end{proof}

A family of connected, cospectral, Cartesian prime graphs automatically satisfies the condition.

\section{Conclusion}

Compared to the Godsil--McKay (GM) switching method, which involves an NP-hard condition verification step (i.e., checking for the existence of a suitable partition) and provides no guarantee of non-isomorphism among the resulting cospectral graphs. The proposed method offers several advantages. Specifically, it involves two efficiently verifiable conditions (Conditions~1 and~2), and one condition (Condition~3) that can be verified using a quasi-polynomial time algorithm for graph isomorphism~\cite{Babai}. Furthermore, our method guarantees that the resulting graphs are mutually non-isomorphic. For a linear-time algorithm to compute the Cartesian prime factorization of a graph, see~\cite{ImrichPeterin}.

The main drawbacks of the proposed method are the rapid increase in graph size, the rarity of connected cospectral graph families of a given size, and the need to compare the prime factorizations of candidate graphs.

GM switching has the advantage of starting with any regular graph and modifying it through vertex additions to avoid the complexity of computing an appropriate partition. A similar idea can be applied to our method: given a family \( \mathcal{G} = \{G_1, G_2, \ldots, G_p\} \) of connected, mutually cospectral graphs, one may choose a singleton family \( \mathcal{H} = \{H_1\} \) such that \( H_1 \) is connected and \( \gcd(|V(G_1)|, |V(H_1)|) = 1 \). The limitation, however, lies in the relative rarity of such families of size \( p \), especially when compared to the abundance of regular graphs available for GM switching.

\subsection*{Acknowledgments}
This work was partially supported by the Department of Science and Technology (Government of India) under the SERB Project (project number SRG/2022/002219) and DST INSPIRE program (grant number Inspire 16/2020).

\end{document}